\documentclass[11pt]{article}
\usepackage{style}

\begin{document}

\title{Quantum algorithms for computational geometry problems\thanks{Supported by the project "Quantum algorithms: from complexity theory to experiment" funded under ERDF programme 1.1.1.5.}}
\author[1]{Andris Ambainis}
\author[1]{Nikita Larka} 
\affil[1]{Center for Quantum Computer Science, Faculty of Computing, University of Latvia}
\date{}
\maketitle

\begin{abstract}
We study quantum algorithms for problems in computational geometry, such as \textsc{Point-On-3-Lines} problem. In this problem, we are given a set of lines and we are asked to find a point that lies on at least $3$ of these lines. \textsc{Point-On-3-Lines} and many other computational geometry problems are known to be \textsc{3Sum-Hard}. That is, solving them classically requires time $\Omega(n^{2-o(1)})$, unless there is faster algorithm for the well known \textsc{3Sum} problem (in which we are given a set $S$ of $n$ integers and have to determine if there are $a, b, c \in S$ such that $a + b + c = 0$).  \\
Quantumly, \textsc{3Sum} can be solved in time $O(n \log n)$ using Grover's quantum search algorithm. This leads to a question: can we solve \textsc{Point-On-3-Lines} and other \textsc{3Sum-Hard} problems in $O(n^c)$ time quantumly, for $c<2$? \\
We answer this question affirmatively, by constructing a quantum algorithm that solves \textsc{Point-On-3-Lines} in time $O(n^{1 + o(1)})$. The algorithm combines recursive use of amplitude amplification with geometrical ideas. We show that the same ideas give $O(n^{1 + o(1)})$ time algorithm for many \textsc{3Sum-Hard}  geometrical problems.
\end{abstract}


\section{Introduction}

The \textsc{3Sum} problem is as follows: given a set of numbers $S$, do there exist $a, b, c \in S$ such that $a + b + c = 0$? There is a nearly trivial classical algorithm that solves this problem in time $O(n^2)$. More advanced algorithms give only a logarithmic improvement to this quadratic complexity \cite{DBLP:journals/corr/JorgensenP14}. It is conjectured that no classical algorithm can solve \textsc{3Sum} problem  in $O(n^{2-\epsilon})$ time. \\
Many problems in computational geometry (for example, determining whether a given set of points contains $3$ points that lie on a line) also seem to require $\Omega(n^2)$ time classically. 
Gajentaan and Overmars \cite{DBLP:journals/comgeo/GajentaanO95} showed that the \textsc{3Sum} problem can be embedded 
into them. This implies that they cannot be solved in $O(n^{2-\epsilon})$ time, unless the \textsc{3Sum} problem can also be solved in 
$O(n^{2-\epsilon})$ time. Such problems are called \textsc{3Sum-Hard}. Besides 3 points on a line, examples of \textsc{3Sum-Hard} problems include determining whether a given set of points contain $3$ points that lie on a line, determining whether a given set of triangles covers given polygon, and determining whether a given set of axis-parallel segments are separable with a line into two nonempty subsets \cite{DBLP:journals/comgeo/GajentaanO95}.  \\
Quantum computing allows designing quantum algorithms that outperform classical algorithms. One such example is Grover search \cite{Grover:1996:FQM:237814.237866}  which achieves quadratic speedup over classical exhaustive search and can be used as a subroutine to speedup more complicated problems \cite{Ambainis:2019:QSE:3310435.3310542, Magniez:2005:QAT:1070432.1070591}. 

In particular, the \textsc{3Sum} problem can be solved by a quantum algorithm in $O(n\log n)$ time, by a fairly simple application of Grover search procedure. Indeed, we can do an exhaustive search over pairs $a, b \in S$ and look for $-(a + b) \in S$ using some data structure (for example, we can use a balanced search tree). However, a direct application of Grover search does not give a quadratic speedup for many geometrical \textsc{3Sum-Hard} class problems. For example, if we need to determine whether a set of points contain three points that lie on the same line, we need to search for all possible triplets of points, which results in $O(n^\frac{3}{2})$ time quantum algorithm \cite{DBLP:journals/qic/Furrow08}.\\
In this paper we combine quantum effects with more sophisticated geometric techniques to design a quantum algorithm with complexity $O(n^{1 + o(1)})$ for \textsc{Point-On-3-Lines} problem. We use ideas from this algorithm to solve many other \textsc{3Sum-Hard}  problems in time $O(n^{1+o(1)})$.

{\bf Related work.} The he \textsc{3Sum} problem has been studied in the context of query complexity and it can be solved with $O(n^{3/4})$ queries, as it is a special case of the subset finding problem of Childs and  Eisenberg \cite{DBLP:journals/qic/ChildsE05} in which one has to find constant-size subset $S$ of an $n$-element set, with $S$ satisfying a certain property. A matching 
$\Omega(n^{3/4})$ quantum query lower bound is known \cite{DBLP:conf/innovations/BelovsS13}. However, the subset finding algorithm of \cite{DBLP:journals/qic/ChildsE05} does not have a time-efficient implementation in the general case. Some special cases
(for example, the element distinctness and $k$-distinctness algorithms of \cite{DBLP:journals/siamcomp/Ambainis07}) can be implemented efficiently but no efficient implementation is known for the 3-SUM case. 

We think that it is unlikely that this line of work would lead to an $o(n)$ time quantum algorithm for the \textsc{3Sum} problem. The element distinctness algorithm \cite{DBLP:journals/siamcomp/Ambainis07} and the subset finding algorithm  \cite{DBLP:journals/qic/ChildsE05} are special cases of a quadratic speedup for hitting times of Markov chains \cite{Szegedy04, abs-1903-07493, abs-1912-04233}. It is unlikely that these methods will lead to a quantum algorithm that is more than quadratically faster than the best classical algorithm.

More generally, we conjecture that the \text{3Sum} problem cannot be solved in $O(n^{1-\epsilon})$ 
quantum time in the QRAM model, neither with methods based on subset finding nor any oher approach.
This could serve as a basis for a quantum version of fine-grained complexity, similarly to recent quantum fine grained 
lower bounds of \cite{abs-1911-01973, abs-1911-05686} based on quantum versions of
Strong Exponential Time Hypothesis (SETH).


\section{Preliminaries}

\subsection{Problems}
Here we define problems we focus on in this paper. All of them belong to \textsc{3Sum-Hard} class. \cite{DBLP:journals/comgeo/GajentaanO95} 

\begin{itemize}
\item \textsc{Point-On-3-Lines}: Given a set of lines in the plane, is there a point that lies on at least three of them? \cite{DBLP:journals/comgeo/GajentaanO95} 
\item \textsc{3-Points-On-Line}: Given a a set of points in the plane, is there a line that contains at least three points? \cite{DBLP:journals/comgeo/GajentaanO95}
\item \textsc{Strips-Cover-Box}: Given a set of strips in the plane (strip is defined as an infinite area between two parallel lines (see Figure \ref{strip})), does their union contain a given axis-parallel rectangle? \cite{DBLP:journals/comgeo/GajentaanO95}
\item \textsc{Triangles-Cover-Triangle}: Given a set of triangles in the plane, does their union contain another given triangle? \cite{DBLP:journals/comgeo/GajentaanO95}
\item \textsc{Point-Covering}: Given a set of $n$ half-planes and a number $t$, determine whether there is a point that is covered by at least $t$ half-planes. \cite{DBLP:journals/comgeo/GajentaanO95}
\item \textsc{Segment-Separator}: Given a set of vertical line segments, does there exists a non vertical line that does not intersect any of given segments and contaisn at least one given segment in each of two half-planes? \cite{DBLP:journals/comgeo/GajentaanO95}
\item \textsc{Visibility-Between-Segments}: Given a set of $n$ vertical line segments $S$ and two particular line segments $s_{1}$ and $s_{2}$, determine whether there is a point on $s_{1}$ and a point on $s_{2}$, such that segment between these two points doesn't intersect any segment from $S$. \cite{DBLP:journals/comgeo/GajentaanO95}
\end{itemize}

We also define \textsc{General-Covering} problem. We will design quantum algorithm for this problem with $O(n^{1+o(1)})$ complexity and then reduce many \textsc{3Sum-Hard} problems to this problem.   

\begin{itemize}
\item \textsc{General-Covering}: We are given a set of $n$ strips and angles (angle is defined as an infinite area between two non-parallel lines (see Figure \ref{angle})) in the plane. The task is to find a point $X$ that satisfies the following conditions:
\begin{itemize}
\item
the point $X$ is an intersection of two angle or strip boundary lines $\ell_{1}, \ell_{2}$ ($\ell_{1}$ and $\ell_{2}$ may be boundary lines of two different angles/strips);
\item 
the point $X$  does not belong to the interior of any angle or strip;
\item
the point $X$ satisfies a given predicate $P(X)$ that can be computed in $O(1)$ time.
\end{itemize}
\end{itemize}

\begin{figure}[!htb]
\minipage{0.5\textwidth}
\centering
\includegraphics[scale=0.33]{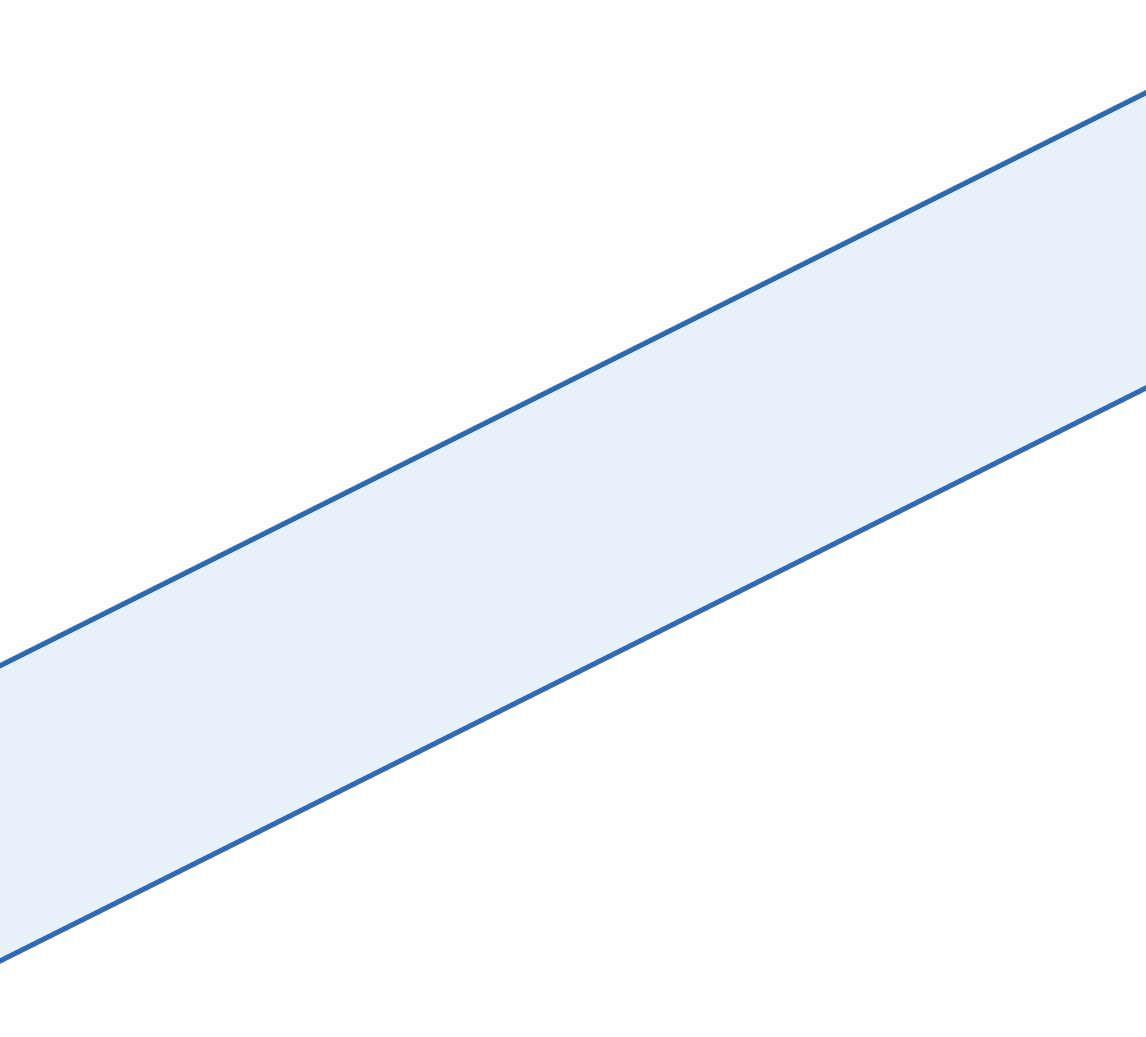}
\caption{Strip}
\label{strip}
\endminipage\hfill
\minipage{0.5\textwidth}
\centering
\includegraphics[scale=0.33]{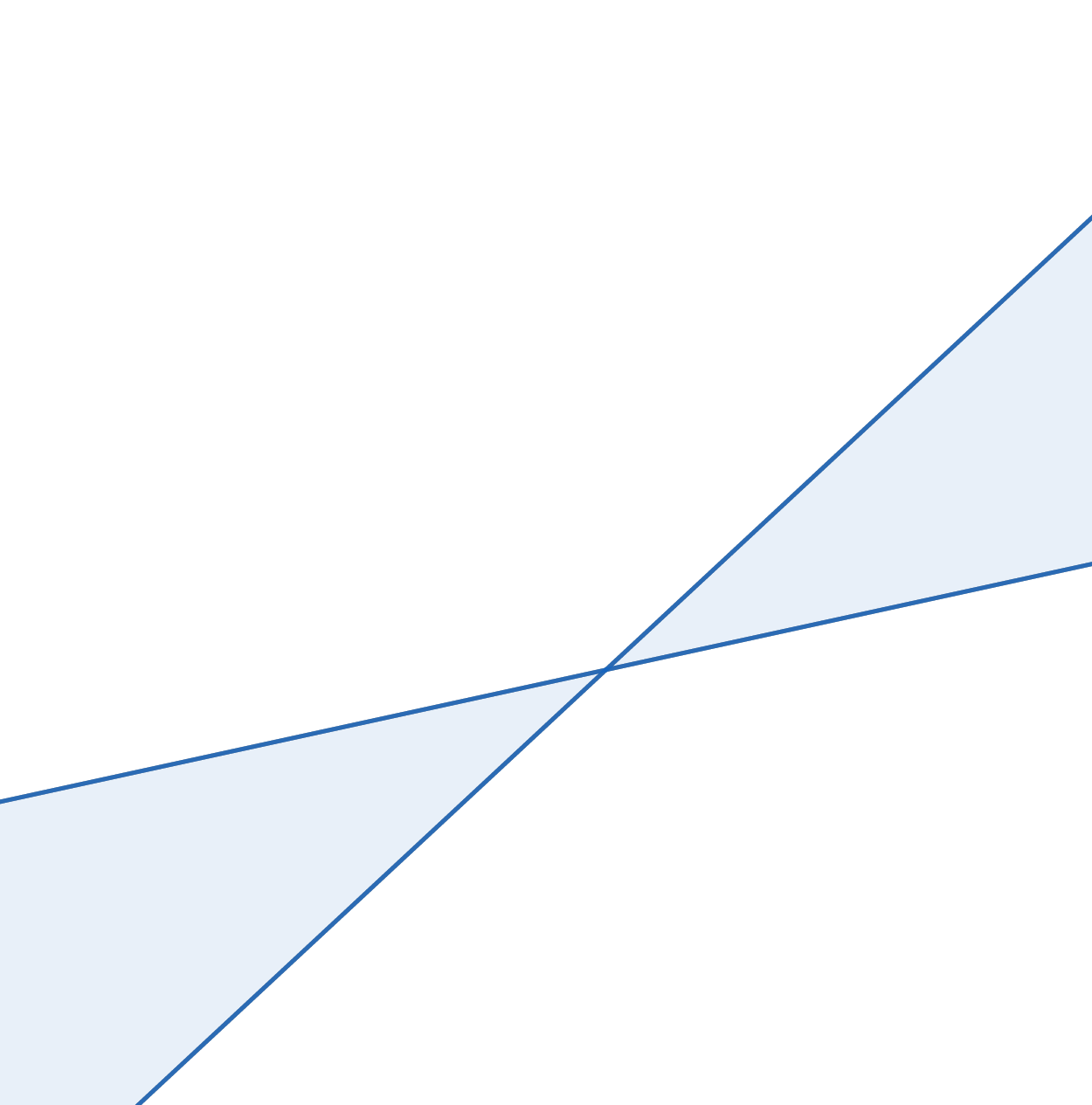}
\caption{Angle}
\label{angle}
\endminipage\hfill
\end{figure}

\subsection{Model}

We assume a query model in which the query returns the description $D_i$ 
of the $i^{\rm th}$ object (point, line, strip, triangle, etc.), given $i$. 
The description consists of several numbers that specify the $i^{\rm th}$ object (e.g. coordinates of a point or values of coefficients
in the equation that specifies a line).  In the quantum case, we can query superpositions of indices $i$. 
The input to quantum query $Q$ consists of two registers, with one register holding $i$ and 
the other register provides the space for $D_i$. The query transformation acts as $Q\ket{i, x}=\ket{i, x\oplus D_i}$.
In particular, given a superposition $\ket{\psi} = \sum_i \alpha_i \ket{i, 0}$ in which $x=0$, 
applying $Q$ gives the state $Q\ket{\psi} = \sum_i \alpha_i \ket{i, D_i}$.
Applying $Q$ to $\ket{\phi} = \sum_i \alpha_i \ket{i, D_i}$ gives the state
$Q\ket{\phi} = \sum_i \alpha_i \ket{i, 0}$ in which the descriptions $D_i$ are erased from the second register.

Our algorithms work in the commonly used QRAM (quantum random access memory) model of computation \cite{QRAM} which assumes quantum memory can be accessed in a superposition. QRAM has the property that any time-$T$ classical algorithm that uses random access memory can be invoked as a subroutine for a quantum algorithm in time $O(T)$. We can thus use primitives for quantum search (e.g., Grover’s quantum search or Amplitude amplification) with conditions checking which requires data stored in a random access memory.

\subsection{Tools}
We will use two well known quantum procedures in our algorithm. 

\begin{theorem}
(Grover search) \cite{Grover_Framework} Given a set of $n$ elements $X = \{x_{1}, x_{2},...,x_{n}\}$ and a boolean function $f : X \rightarrow \{0, 1\}$. The task is to find $x \in X$ such that $f(x) = 1$. There is a bounded-error quantum procedure that solves this problem using $O(\sqrt{n})$ quantum queries.
\end{theorem}

\begin{theorem}
(Amplitude amplification) \cite{Amplitude_Amplification} Let $A$ be a quantum procedure with one-sided error and success probability at least $\epsilon$. Then, there is a quantum procedure $B$ that solves the same problem with success probability $\frac{2}{3}$ invoking $A$ $O(\frac{1}{\sqrt{\epsilon}})$ times.
\end{theorem}
Note that any constant success probability $1-\epsilon$ can be achieved with repeating Amplitude Amplification constantly many times. \\

We will also use the following well known computational geometry algorithm. 

\begin{theorem}
(Arrangement of lines) \cite{Lines_Arrangement} 
\label{thm:arrange}
Given a set of $n$ lines in the plane, we can compute partition of the plane formed by those lines in time $O(n^2)$. 
\end{theorem}


We will also use point-line dualization for problem reductions. Point-line dualization is a plane transformation that maps points to lines and lines to points in the following way:
\begin{itemize}
\item Line $\ell: y=ax + b$ is mapped to point $\ell^{*}=(a,-b)$
\item Point $P = (a, b)$ is mapped to line $P^{*}: y=ax-b$
\end{itemize}
One may note, that the following properties are true:
\begin{enumerate}
\item $(P^*)^* = P$ and $(\ell^*)^* = \ell$
\item $P \in \ell \iff \ell^* \in P^*$
\item If point $A, B, C$ lie on one non-vertical line, then lines $A^*, B^*, C^*$ meet at one point.
\item If non-vertical lines $p, q, r$ meet at one point, then points $p^*, q^*, r^*$ lie on one line.
\item Points from vertical line segment are mapped to a strip.
\item Points from non-vertical line segment are mapped to an angle.
\end{enumerate}


\section{Point on three lines}
In this part we describe a quantum algorithm which solves \textsc{Point-On-3-Lines} problem in $O(n^{1+o(1)})$ time, improving over the $O(n^{3/2})$ time quantum algorithm of Furrow \cite{DBLP:journals/qic/Furrow08}. The idea behind our algorithm is as follows. Suppose we are given a set $S$ of lines in the plane. From this set we randomly pick $k$ lines. Those $k$ lines split plane into no more than $\binom{k + 1}{2} + 1$ regions. Each line from set $S$ intersects $k + 1$ regions. So, on average, a region contains $O(\frac{n}{k})$ lines crossing that region. \\
These facts give an opportunity to design a quantum algorithm that improves over the complexity $O(n^\frac{3}{2})$ of simple quantum search (as in Furrow's algorithm \cite{DBLP:journals/qic/Furrow08}).
 If we have a point $A$, then we can decide if $A$ belongs to some line from $S$ using $O(\sqrt{\frac{n}{k}} + k)$ quantum time. To do so, we need to find a region which contains point $A$ and check if $A$ belongs to a line which crosses that region. For finding the region, $O(k)$ time suffices. Checking if $A$ belongs to a line from $S$ can be done by Grover's search
over $O(\frac{n}{k})$ lines that cross this region. For this, $O(\sqrt{\frac{n}{k}})$ time suffices.\\
If we need to find three lines that intersect in one point, we run Grover search over all pairs of lines $(\ell_{i}, \ell_{j})$. For each pair, we find the intersection point $P_{i,j} = l_{i} \cap l_{j}$ and check if point $P_{i,j}$ belongs to a third line using algorithm described earlier. If the subdivision of the plane into regions can be done in time  $O(nk^2)$, this algorithm runs in time $O(nk^2 + \sqrt{n^2}(k + \sqrt{\frac{n}{k}}))$. Setting $k=n^\frac{1}{5}$ gives $O(n^\frac{7}{5})$. But we can find even better algorithm. After dividing the plane into regions, instead of searching for an intersection point of three lines, we search for a region which has this point (search is done using Grover search) and we recursively apply $O(n^\frac{7}{5})$ algorithm to find the intersection point of three lines inside that region. 
We can then add more levels of recursion to decrease the complexity further. We now describe the final algorithm (in which we recurse at the optimal choice of $k$ and the number of levels of recursion grows with $n$). \\
Let $S$ be the given set of lines in a plane. Let $P$ be a subset of $S$ containing exactly $k$ lines. Lines in $P$ divide plane into convex (possibly infinite) polygons. We arbitrarily triangulate regions, which are bounded by at least $4$ lines. This results in a subdivision of the plane  into regions $R_{1}, R_{2},...,R_{t}$ where each region is bounded by at most $3$ segments (see Figure \ref{separation}). Let $s(R_{i}) = \{\ell \mid \ell \in S \text{ and } \ell \text{ intersect } R_{i}\}$. \\

\begin{figure}[h]
\centering
\includegraphics[scale=0.4]{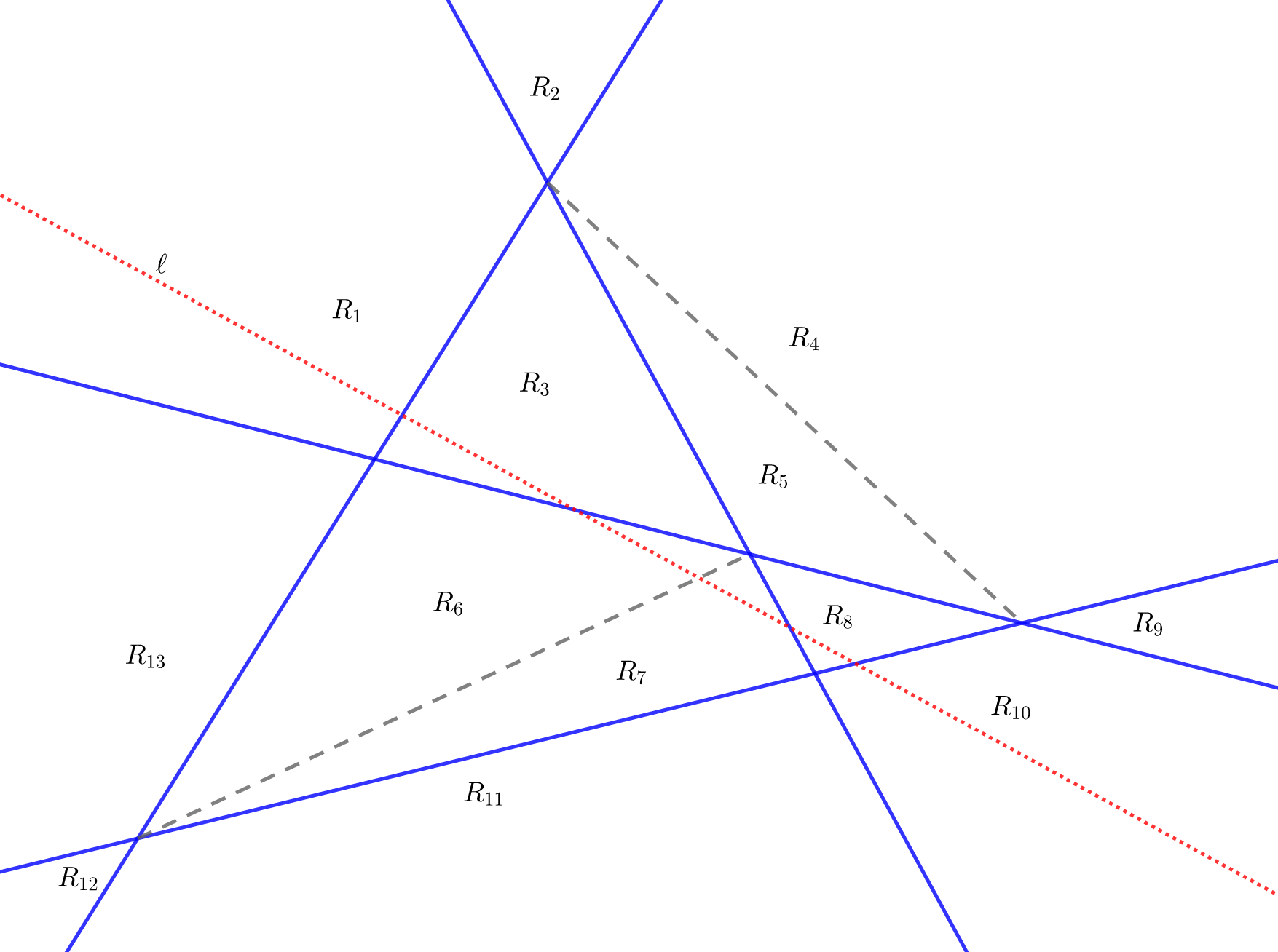}
\caption{Blue lines divide plate into 13 regions: $R_{1}, R_{2}, ...,R_{13}$. Red line $\ell$ passes through regions: $\ell \in s(R_{1}),  s(R_{3}),  s(R_{6}),  s(R_{7}),  s(R_{8}),  s(R_{10})$}
\label{separation}
\end{figure}

We start with the following three observations: 

\begin{lemma}
If after the triangulation we get $t$ regions: $R_{1}, R_{2},...,R_{t}$, then $t \leq 2|P|^2 = 2k^2$
\end{lemma}

\begin{proof}
If $f_{i}$ is the number of faces bounded by $i$ lines before the triangulation, then
\begin{flalign}
& t = f_{1} + f_{2} + f_{3} + 2f_{4} + 3f_{5} + ... + (k-2)f_{k} \leq \sum_{i=1}^{k}if_{i} \leq 2k^2
\end{flalign}
The last inequality holds because $\sum_{i=1}^{k}if_{i}$ is equal to twice the total number of line segments before the triangulation (because each line segment is on the boundary of two faces, one on each side) and the number of line segments is at most $k^2$ (because each of $k$ lines is split by the other $k-1$ lines into at most $k$ segments).
\end{proof}

\begin{lemma}
Sets $s(R_{1}),...,s(R_{t})$ can be built classically in time $O(|S| \times |P|^2)$.
\end{lemma}

\begin{proof}
We can construct regions $R_{1}, R_{2},...,R_{t}$ in time $O(|P|^2)$ using Theorem \ref{thm:arrange}. 
We build $s(R_{i})$ by iterating through each line from $S$ and checking whether the line intersects region $R_{i}$. This step takes time $O(|S| \times t) = O(|S| \times |P|^2)$.
\end{proof}

\begin{lemma}
If $P$ is chosen uniformly at random, then
\begin{equation}
\Pr[\max_{i}|s(R_{i})| \, \ge \, 3\frac{|S|}{|P|}(5\log(|S|) + \log(\epsilon^{-1}))] \, \le \, \epsilon
\end{equation} 
\end{lemma}

\begin{proof}
Let $\ell$ be an arbitrary line (possibly not from the set $S$). Lines from set $S$ intersect line $\ell$ in points $X_{1}, X_{2},...,X_{m}$ in this order (if two lines $\ell_{i}$ and $\ell_{j}$ intersect $\ell$ in the same point, then $X_{i} = X_{j}$). We note that $m \leq |S|$, since some lines from $S$ might be parallel to $\ell$. We color a point $X_{i}$ with white color if the corresponding line $\ell_{i}$ from $S$ is in the set $P$. Otherwise we color the point $X_{i}$ with black color. We define $L = \lc \frac{|S|}{|P|}(5\log(|S|) + \log(\epsilon^{-1})) \rc$ and we assume that $L \leq |S|$, since otherwise lemma is obviously true. 

We say that a line $\ell$ is bad, if there exists index $i$, such that $X_{i+j}$ is black for all $j \in [0...L-1]$.  
The probability of  $\ell$ being bad can be upper bounded as follows

\begin{flalign}
&\Pr\Big[\bigvee_{i=1}^{m-L+1}(X_{i}, X_{i+1},...,X_{i + L - 1} \text{ are all black})\Big] = \ \\
&\leq \sum_{i=1}^{m - L + 1}\Pr[X_{i}, X_{i+1},...,X_{i + L - 1} \text{ are all black}]
= (m - L + 1) \frac{{|S| - L\choose |P|}}{{|S|\choose|P|}}  \\ 
&\leq |S|\bigg(\frac{|S| - L}{|S|}\bigg)^{|P|}
\leq |S|\Bigg[\bigg(1 - \frac{L}{|S|}\bigg)^{\frac{|S|}{L}}\Bigg]^{L\frac{|P|}{|S|}} \leq \frac{|S|}{e^{5\log|S| + \log{\epsilon^{-1}}}} = \frac{\epsilon}{|S|^4}
\end{flalign}

Consider set $S^\prime$ which consists of lines from $S$ and 
lines that pass through at least two intersection points of lines from $S$. 
Then, every edge $e$ of every region $R_i$ lies on a  line that belongs to $S'$
(because $e$ is either a segment of one of original lines from $S$ or is created during the triangulation and,
in the second case, both endpoints of $e$ are intersection points of two lines from $S$)  
Since there are at most ${|S| \choose 2}$ intersection points of lines from $S$, we have
  $|S^\prime| \leq |S|^4$ and

\begin{equation}
\Pr[S^\prime \text{ contains bad line}] \leq \sum_{\ell \in S^\prime}\Pr[\ell \text{ is bad}] \leq |S^\prime| \times \frac{\epsilon}{|S^4|} \leq \epsilon
\end{equation}

To finish the proof, it is enough to see that the fact that $S^\prime$ doesn't contain bad line implies $|s(R_{i})| \leq 3(L-1) < 3\frac{|S|}{|P|}(5\log(|S|) + \log(\epsilon^{-1}))$ for all $i$. Indeed, if no lines in $S^\prime$ is bad, then each side of each region $R_{i}$ contains less than $L$ black points. Since a black point corresponds to a line which intersects a region $R_{i}$ and each region is bounded by at most three segments, the region $R_{i}$ is intersected by no more than $3(L-1)$ lines from $S$. 

\end{proof}

\begin{theorem}
There is a bounded error quantum algorithm for \textsc{Point-On-3-Lines} problem, that runs in time $O(|S|^{1 + o(1)})$.

\begin{proof}
The algorithm has a parameter $k$ and allowable error probability $\epsilon$. The algorithm consists of a recursive procedure that takes a set of lines $X$ as input and returns $3$ lines from the set $X$ which intersects at one point or tells that there are no such $3$ lines. 

\begin{procedure}
\caption{$Algo_{k}$($X$)}
\If{$|X| < k$}{
Check for an intersection of $3$ lines classically, by exhaustive search
}
$R_{1}, R_{2},...,R_{t} = \text{RandomPlaneSeparation}_{\frac{\epsilon}{2}}(X)$ \\
$\text{Build sets: } s(R_{1}), s(R_{2}),...,s(R_{t})$ \\
\If{$max_{i}|s(R_{i})| > 3\frac{|X|}{k}(5\log(|X|) + \log(\frac{2}{\epsilon}))$}{
return error
}
Let $A$ be the algorithm that randomly chooses $j\in[t]$ and runs $Algo_{k}(s(R_{j}))$. \\
With Amplitude amplification, run $A$ with the success probability amplified to at least $1-\frac{\epsilon}{2}$. \\

\end{procedure}

The recursive procedure can be described as follows. If the input set $X$ contains less than $k$ lines, we solve \textsc{Point-On-3-Lines} classically in time $O(|X|^2)$. Otherwise, we split the plane into regions $R_{1}, R_{2}...,R_{t}$ with $k$ random lines from $X$ and build sets $s(R_{1}), s(R_{2}),...,s(R_{t})$. If there are three lines that intersect at one point, then this point is located in one of the regions $R_{i}$ (if this point is on the boundary of a region, then this point can be found during the construction of sets $s(R_{i})$). We use amplitude amplification to find the region $R_{i}$ which contains intersection point of three lines and recursively apply the procedure to lines from $s(R_{i})$. \\

Suppose that $S$ contains three lines that intersect at one point. If $|S| < k$, the algorithm finds those $3$ lines with certainty. 
Let $|S| \geq k$. The probability that $s(R_{i}) > 3\frac{|S|}{k}(5\log(|S|) + \log(\frac{2}{\epsilon}))$ for some $R_i$ and the algorithm returns an error is less than $\frac{\epsilon}{2}$. 
By executing amplitude amplification (running $Algo_{k}(s(R_{i}))$  $O(\sqrt{t}) = O(k)$ times), we can reduce probability of not finding the $3$ lines to $\frac{\epsilon}{2}$. 
So, $Algo_k(S)$ finds desired three lines with probability at least $1 - \epsilon$. \\

If $T(|X|)$ is a runtime of $Algo_k(X)$, then:
\begin{equation}
T(|X|) = O(|X| \times k^2) + O\big(\sqrt{k^2}\big) \times T\Big(3\frac{|X|}{k}(5\log(|X|) + \log(\frac{2}{\epsilon}))\Big)
\end{equation} 
If $ k = |S|^\frac{1}{\alpha} \cdot 3(5\log(|S|) + \log(\frac{2}{\epsilon}))$, then:
\begin{equation}
T(|X|) \leq O(|X| \times |S|^\frac{2}{\alpha}\log^2(|S|)) + O\big(\sqrt{k^2}\big) \times T\big(|X| \times |S|^{-\frac{1}{\alpha}}\big)
\end{equation}
There are $(C_{1}k)^{2j}$ problems on recursion level $j$ for some constant $C_{1}$ and each problem size is at most $|S|^{1 - \frac{j}{\alpha}}$.
\begin{equation}
\begin{aligned}
T(|S|) \leq \sum_{j=0}^{\alpha} \sqrt{(C_{1}k)^{2j}} \times \Big[|S|^{1 - \frac{j}{\alpha}} \times |S|^\frac{2}{\alpha}\log^2(|S|)\Big] = \\
= \sum_{j=0}^{\alpha} \Bigg(\frac{C_{1}k}{|S|^{\frac{1}{\alpha}}}\Bigg)^j \big(|S|^{1 + \frac{2}{\alpha}}\log^2(|S|)\big) \leq \\
\leq \alpha \big(C_{2} \log(|S|)\big)^\alpha \big(|S|^{1 + \frac{2}{\alpha}} \log^2(|S|)\big)
\end{aligned}
\end{equation}
If $\alpha = \sqrt{\frac{2\log(|S|)}{\log(C_{2}) + \log\log|S|}}$, then $T(|S|) = O(\alpha |S|^{1 + \frac{4}{\alpha}} \log^2(|S|)) = O(|S|^{1 + o(1)})$ 
 
\end{proof}

\end{theorem}


\section{Other 3SUM hard problems}

In this section, we show how it is possible to apply plane separation ideas, described in the previous section, to speed up other \textsc{3-Sum-Hard} problems defined in \cite{DBLP:journals/comgeo/GajentaanO95}.
\begin{itemize}
\item \textsc{3-Points-On-Line} \\ 
This problem is dual to the \textsc{Point-On-3-Lines} problem \cite{DBLP:journals/comgeo/GajentaanO95}, and so is solvable with the same quantum algorithm in time $O(n^{1+o(1)})$, as described in the previous section.
\item \textsc{General-Covering} \\
The given $n$ strips and angles form $2n$ lines. We divide the plane into regions by randomly choosing $k$ out of those $2n$ lines, similarly to the algorithm described in the previous section. A region and a strip/angle can be in one of the following relations: the strip/angle fully covers the region, the strip/angle partly covers the region or the strip/angle has no common points with the region. We can identify all regions that are fully covered by some strip/angle in time $O(nk^2)$. For regions that are not fully covered by some strip/angle, we identify the set $s(R_{i})$ of strips and angles which cross that region. Non covered regions may contain the desired intersection point, but this intersection point is formed by the lines that are boundary lines of the strips/angles in the set $s(R_{i})$. Similarly to the algorithm \textsc{Point-On-3-Lines}, our task is divided into $O(k^2)$ tasks, each of which involved $O(\frac{n\log n}{k})$ objects. 
The time complexity is
\begin{equation}
T(n) = O(nk^2) + O\big(\sqrt{k^2}\big) \times T\Big(3\frac{2n}{k}(5\log(2n) + \log(\frac{2}{\epsilon}))\Big) .
\end{equation}
Similarly to the analysis of \textsc{Point-On-3-Lines} problem, we get $T(n)=O(n^{1 + o(1)})$.
\item \textsc{Strips-Cover-Box} \\
This problem is just the special case of the \textsc{General-Covering} problem with the predicate $P(X)$ being true if the point $X$ is located inside the given box. Then, the point $X$ from \textsc{General-Covering} problem corresponds to an uncovered point in \textsc{Strips-Cover-Box} problem. So, \textsc{Strips-Cover-Box} can also be solved in time $O(n^{1 + o(1)})$.
\item \textsc{Triangles-Cover-Triangle} \\
The given $n$ triangles consist of $3n$ segments. We extend each segment to a line and separate the plane into regions, similarly to the \textsc{Point-On-3-Lines} problem with randomly chosen $k$ lines. A triangle and a region can be in one of the following relations: the triangle fully covers the region, the triangle partly covers the region or the triangle has no common point with the region. We can identify all regions that are fully covered by some triangle in time $O(nk^2)$. All other regions may contain a point $X$ which is not covered by any triangle. Similarly to the \textsc{Point-On-3-Lines} problem, we search for the region $R_{i}$ which contains that point. Note that, if a triangle partly covers the region $R_{i}$, then at least one of the segments that form this triangle is in $s(R_{i})$. So, we can finish our algorithm, just as in \textsc{General-Covering} problem, with the predicate $P(X)$ being true, if the point $X$ is located inside the triangle that must be covered.     
\item \textsc{Point-Covering} \\
The given $n$ half-planes are specified by $n$ lines. We separate the plane into regions, similarly to the \textsc{Point-On-3-Lines} problem, by randomly choosing $k$ out of $n$ given lines. For each region $R_{i}$ we compute the number of half-planes $r_{i}$ that fully cover this region. This takes $O(nk^2)$ time. To determine if there exists a point that is covered by at least $t$ half-planes, we need to tell, if there exists a point inside a region $R_{i}$ that is covered by at least $t - r_{i}$ half-planes from $s(R_{i})$. As in the \textsc{General-Covering} problem, the algorithm takes $O(n^{1 + o(1)})$ time.
\item \textsc{Visibility-Between-Segments} \\
We dualize the given $n$ vertical segments, to get $n$ strips. We need to find a point that does not belong to any of the strips and has the property that the corresponding  line in the initial plane intersects two given segments $s_{1}$ and $s_{2}$. This problem is just the special case of the \textsc{General-Covering} problem, where the predicate $P(X)$ is true if the line corresponding to $X$ intersects two given segments $s_{1}$ and $s_{2}$. Just like in the \textsc{General-Covering} problem, this problem can also be solved in $O(n^{1 + o(1)})$ time.
\item \textsc{Segment-Separator} \\
This problem can be solved in exactly the same way as the \textsc{Visibility-Between-Segments} problem, with the only difference being the predicate $P(X)$. Now this predicate is true, if the line corresponding to a point separates the given segments in the way required for the  \textsc{Segment-Separator} problem. Since $X$ is the intersection point of two lines in the dual plane, the line, corresponding to the point $X$, must go through two endpoints of two different given segments. So, $P(X)$ is false, if the corresponding line goes through an edge of the convex hull of the endpoints of given segments. This can be detemined in time $O(1)$, after we precompute the convex hull in time $O(n\log(n))$. This results in an $O(n^{1 + o(1)})$ time quantum algorithm.
\end{itemize}


\bibliographystyle{alpha}
\bibliography{bibliography}

\end{document}